\documentclass[reqno,centertags,11pt]{amsart}
\usepackage{amsmath,amsthm,amscd,amssymb} \usepackage{latexsym}
\usepackage{graphicx}
\usepackage{color}

 \newcommand{\R}{{\mathbb{R}}}
 \newcommand{\C}{{\mathbb{C}}}
 
 \newcommand{\Z}{{\mathbb{Z}}}

\newcommand{\beq}{\begin{equation}}
\newcommand{\eeq}{\end{equation}}
\newcommand{\bdm}{\begin{displaymath}}
\newcommand{\edm}{\end{displaymath}} \newcommand{\ba}{\begin{align}}
\newcommand{\ea}{\end{align}} \newcommand{\bpf}{\begin{proof}}
\newcommand{\epf}{\end{proof}}

 \allowdisplaybreaks

\DeclareMathOperator{\Tr}{Tr}

\newtheorem{theorem}{Theorem}[section]
\newtheorem{lemma}[theorem]{Lemma}
\newtheorem{corollary}[theorem]{Corollary}

\theoremstyle{definition}
\newtheorem{definition}[theorem]{Definition}
\theoremstyle{remark}
\newtheorem{remark}[theorem]{Remark}

\begin{document}

\title[Solutions of NPE with Periodic Potential]{ Solutions of Nonlinear Polyharmonic Equation with Periodic Potential}
\author[Yu.~Karpeshina, S.-K. Kim]{Yulia~Karpeshina, Seong-Uk~Kim}
%


\address{Department of Mathematics, Campbell Hall, University of Alabama at Birmingham,
1300 University Boulevard, Birmingham, AL 35294.}
\email{karpeshi@uab.edu}%

\address{Department of Mathematics, Julian Science and Math Center,
Depauw University,
Greencastle, IN 46135}%
\email{ksw8755@gmail.com}

\address{}
\email{}

\thanks{Supported in part by NSF-grants DMS-1201048 (Y.K.) }
\thanks{The authors are thankful to Professor Roman Shterenberg (UAB) for fruitful discussions.}

\date{\today}


\maketitle
\begin{abstract} Quasi-periodic solutions of  a nonlinear periodic polyharmonic equation in $\R^n$, $n>1$, are studied. It is proven that there is an extensive "non-resonant"  set ${\mathcal G}\subset \R^n$ such that for every $\vec k\in \mathcal G$ there is a  solution  asymptotically close to a plane wave 
$Ae^{i\langle{ \vec{k}, \vec{x} }\rangle}$ as $|\vec k|\to \infty $.\end{abstract}
\section{Introduction}
Let us consider a nonlinear polyharmonic equation with quasi-periodic boundary conditions:
\begin{equation}\label{main equation, 2l>n}
(-\Delta)^{l}u(\vec{x})+V(\vec{x})u(\vec{x})+\sigma |u(\vec{x})|^{2}u(\vec{x})=\lambda u(\vec{x}), ~\vec{x}\in \mathbb{R}^{n},\end{equation}
\begin{equation}\label{main condition, 2l>n}
\begin{cases}
~u(x_{1},\cdots,\underbrace{2\pi}_{s-th},\cdots,x_{n})=e^{2\pi it_{s}}u(x_{1},\cdots,\underbrace{0}_{s-th},\cdots,x_{n}),\\
~\frac{\partial}{\partial x_{s}}u(x_{1},\cdots,\underbrace{2\pi}_{s-th},\cdots,x_{n})=e^{2\pi it_{s}}\frac{\partial}{\partial x_{s}}u(x_{1},\cdots,\underbrace{0}_{s-th},\cdots,x_{n}),\\
\vdots\\
~\frac{\partial^{2l-1}}{\partial x_{s}^{2l-1}}u(x_{1},\cdots,\underbrace{2\pi}_{s-th},\cdots,x_{n})=e^{2\pi it_{s}}\frac{\partial^{2l-1}}{\partial x_{s}^{2l-1}}u(x_{1},\cdots,\underbrace{0}_{s-th},\cdots,x_{n}),\\
~~~~s=1,\cdots,n.
\end{cases}
\end{equation}
where $l$ is an integer,  $n\geq 2$ , $\vec t=(t_1,...,t_n)$ is a parameter (quasimomentum), $\vec t \in K:=[0,1]^{n}$ and $V(\vec{x})$ is a periodic potential with   an elementary cell 
$Q:=[0,2\pi)^{n}$. Restriction on  smoothness of $V(\vec x)$ is given by the inequality:
\begin{equation}\label{potential condition, 2l>n}
\sum_{q \in \mathbb{Z}^{n}}|v_{q}|<\infty,
\end{equation}
$v_q$ being  Fourier coefficients.
Without the loss of generality, we assume  $v_0=0$. 

When $l=1$, $n=1,2,3$,   equation \eqref{main equation, 2l>n} is a famous Gross-Pitaevskii equation for Bose-Einstein condensate., see e.g. \cite{PS08}.
%
%
In  physics papers, e.g. \cite{KS02}, \cite{LO03}, \cite{YB13}, \cite{YD03},  numerical computations for  Gross-Pitaevskii equation are made. 
However, they are restricted to the one dimensional case. There is a lack of theoretical considerations even for the case $n=1$. In this paper we study the case $2l>n$, $ n\geq 2,$ preparing the ground for a physically interesting case $l=1$, $n=2,3$.

The goal of the paper is to construct asymptotic formulas for  $u(\vec{x})$ as $\lambda \to \infty $.
We show that there is an extensive "non-resonant"  set ${\mathcal G}\subset \R^n$ such that for every $\vec k\in \mathcal G$ there is a quasiperiodic solution of (\ref{main equation, 2l>n}) close to a plane wave 
$Ae^{i\langle{ \vec{k}, \vec{x} }\rangle}$ with $\lambda=\lambda (\vec k, A)$ close to $|\vec k|^{2l}+\sigma |A|^2$ as $|\vec k|\to \infty $ (Theorem \ref{main theorem 2l>n}). We assume $A\in \C$,  
\begin{equation}\sigma |A|^2<\lambda ^{\gamma }, \ \ 0<\gamma <(2l-n)/2l, \label{A}
\end{equation}
the quasimomentum $\vec t$ in \eqref{main equation, 2l>n}
being defined by the formula: $\vec k=\vec t+2\pi j$, $j\in \Z^n$.

We show that the non-resonant set $\mathcal G$ has an asymptotically full measure in $\R^n$:
\begin{equation}
\lim _{R\to \infty}\frac{\left| \mathcal G\cap B_R\right|_n}{|B_R|_n}=1, \label{full}
\end{equation}
where $B_R$ is a ball of radius $R$ in $\R^n$ and $|\cdot |_n$ is Lebesgue measure in $\R^n$.

 Moreover,
we investigate a set    $\mathcal{D}_{}(\lambda, A)$ of vectors $\vec k\in \mathcal G$, corresponding to a fixed sufficiently large $\lambda $ and a fixed $A$.   The set $\mathcal{D}_{}(\lambda, A)$,
defined as a level (isoenergetic) set for $\lambda _{}(\vec
k, A)$, \begin{equation} {\mathcal D} _{}(\lambda,A)=\left\{ \vec k \in
\mathcal{G} _{} :\lambda _{}(\vec k, A)=\lambda
\right\},\label{isoset} \end{equation}
is proven to be a slightly distorted $n$-dimensional sphere with with 
a finite number of holes (Theorem \ref{iso}). For any sufficiently large $\lambda $, it can be described by  the formula:
\begin{equation} {\mathcal D}_{}(\lambda, A)=\{\vec k:\vec
k=\varkappa _{}(\lambda, A,\vec{\nu})\vec{\nu},
    \ \vec{\nu} \in {\mathcal B}_{}(\lambda)\}, \label{D}
    \end{equation}
where ${\mathcal B}_{}(\lambda )$ is a subset of the unit
sphere $S_{n-1}$. The set ${\mathcal B}_{}(\lambda )$ can be
interpreted as a set of possible directions of propagation for the
almost plane waves.  The set ${\mathcal B}_{
}(\lambda )$ has  an asymptotically full
measure on $S_{n-1}$ as $\lambda \to \infty $:
    \begin{equation}
\left|{\mathcal B}_{}(\lambda )\right|=_{\lambda \to \infty
}\omega _{n-1} +O\left(\lambda^{-\delta }\right), \ \ \delta >0,\label{B}
    \end{equation}
here $\omega _{n-1} $ is the standard surface measure of  $S_{n-1}$.
The value $\varkappa _{}(\lambda ,A,\vec \nu )$ in (\ref{D}) is
the ``radius" of ${\mathcal D}_{}(\lambda,A)$ in a direction
$\vec \nu $. The function $\varkappa _{}(\lambda ,A,\vec \nu
)-(\lambda-\sigma |A|^2)^{1/2l}$ describes the deviation of ${\mathcal
D}_{}(\lambda,A)$ from the perfect circle of the radius
$(\lambda-\sigma |A|^2)^{1/2l}$. It is proven that the deviation is asymptotically
small:
    \begin{equation} \varkappa _{}(\lambda ,A, \vec \nu
)=_{\lambda \to \infty} \left(\lambda-\sigma |A|^2\right)^{1/2l}+O\left((1+|\sigma ||A|^2)\lambda^{-\gamma }\right),\ \  0<\gamma <(2l-n)/2l.\label{h}
    \end{equation}

To prove the results above, we consider the term $V+\sigma|u|^{2}$ in  equation (\ref{main equation, 2l>n}) as a periodic potential and formally change the nonlinear equation to a linear equation with an unknown potential $V(\vec{x})+\sigma |u(\vec{x})|^{2}$:
\begin{equation*}
(-\Delta)^{l}u(\vec{x})+\big(V(\vec{x})+\sigma |u(\vec{x})|^{2}\big)u(\vec{x})=\lambda u(\vec{x}).
\end{equation*} 
Further, we use  results obtained in \cite{K97} for  linear polyharmonic equations.
To start with, we consider a linear operator in $L^{2}(Q)$ 
described by the formula 
\begin{equation}H(\vec t)=(-\Delta)^{l}+V,\label{linear oper 2l>n}\end{equation}
and  quasi-periodic boundary condition (\ref{main condition, 2l>n}).
The free operator $H_{0}(\vec t)$, corresponding to $V=0$, has eigenfunctions given by:
\begin{equation}\psi_{j}(\vec{x})=e^{i\langle{ \vec{p}_{j}(\vec t), \vec{x} }\rangle},~~\vec{p}_{j}(\vec t):=\vec t+2\pi j,~j \in \mathbb{Z}^{n},~\vec t \in K,\label{0}\end{equation}
and the corresponding eigenvalue is $p_{j}^{2l}(\vec t):=|\vec{p}_{j}(\vec t)|^{2}$. 
Perturbation theory for operator $H(\vec t)$  is developed in 
\cite{K97}. It is shown that at  high energies, there is an extensive set of generalized eigenfunctions  being close to plane waves. Below (See Theorem \ref{linear pert 2l>n}), we describe this result in details.
    
    Next, we define a map ${\mathcal M}: L^{\infty}(Q) \rightarrow L^{\infty}(Q)$ by
\begin{align}\label{def of A 2l>n}
{\mathcal M}W(\vec{x})=V(\vec{x})+\sigma|u_{\tilde{W}}(\vec{x})|^{2}.
\end{align}
Here, $\tilde{W}$ is a shift of $W$ by a constant such that $\int_{Q}\tilde{W}(\vec{x})d\vec{x}=0$,
$$\tilde{W}(\vec{x})=W(\vec{x})-\frac{1}{(2\pi)^{n}}\int_{Q}W(\vec{x})d\vec{x},$$ 
and $u_{\tilde{W}}$  is an eigenfunction of the linear operator, $(-\Delta)^{l}+\tilde{W}$  
with  (\ref{main condition, 2l>n}). 
We consider a sequence  $\{W_{m}\}_{m=0}^{\infty}$:
\begin{align}\label{def of successive sequence A 2l>n}
W_{0}=V+\sigma |A|^2,\ \ \ {\mathcal M}W_{m}=W_{m+1}.
\end{align}
Note that the sequence is well defined, since for each $m=1,2,3,\cdots$ and $\vec t$ in a neighborhood
of a non-resonant set  described in Section 2, there is   an eigenfunction $u_{{m}}(\vec{x})$ corresponding to the potential $\tilde{W}_{m}$:
%
$$H_{{m}}(\vec t)u_{{m}} =\lambda_{{m}}u_{{m}},$$
$$H_{{m}}(\vec t)u_{{m}} :=(-\Delta)^{l}u_{{m}}+\tilde{W}_{m}u_{{m}}.$$
where $\lambda_{{m}}$, $u_{m}$ are as described in Theorem \ref{linear pert 2l>n}. 
Next, we  prove that the sequence $\{W_{m}\}_{m=0}^{\infty}$  is a Cauchy sequence of periodic functions in $Q$ with respect to  a norm 
\begin{align}\label{def of star norm 2l>n}
\|W\|_{*}=\sum_{q\in \mathbb{Z}^{n}}|w_{q}|,
\end{align}
 $w_{q}$ being Fourier coefficients of $W$.
%
This implies that there is a periodic function $W$ such that 
$$W_{m}\rightarrow W,~\mbox{with respect to the norm}~\|\cdot\|_{*}.$$
Further,  we show that 
$$u_{{m}}\rightarrow u_{\tilde W}, ~\mbox{in}~L^{\infty}(Q),$$
%
$$\lambda_{{m}}\rightarrow \lambda_{\tilde W}, ~\mbox{in}~\mathbb{R},$$
where $u_{\tilde W}$, $ \lambda_{\tilde W}$ correspond to the potential ${\tilde W}$ as described in Theorem \ref{linear pert 2l>n}.
%
It follows from (\ref{def of A 2l>n}) and (\ref{def of successive sequence A 2l>n}) that ${\mathcal M}W=W$ and hence $u_{}:=u_{\tilde W}$ solves the nonlinear equation with quasi-periodic boundary condition, (\ref{main equation, 2l>n}) and (\ref{main condition, 2l>n}). 
%

The paper is organized as follows. In Section 2, we describe some known results for the linear operator. They include perturbation formulas for Bloch eigenvalues and the corresponding spectral projections at high energies. In Section 3, we prove the  main lemma (Lemma \ref{main lemma  1 2l>n}), using the perturbation formulas for the linear operator.   Based on the main lemma, we show the existence of a solution  of   (\ref{main equation, 2l>n}) and (\ref{main condition, 2l>n}) close to plane wave (Theorem \ref{main theorem 2l>n}).   Section 4 is devoted to isoenergetic surface $\mathcal D(\lambda, A)$. We  prove formulas \eqref{D} -- \eqref{h}  (Theorem \ref{iso}).

\section{Linear Operator}
Let us consider an operator
\begin{equation}\label{linear oper rn 2l>n}
H=(-\Delta)^{l}+V,
\end{equation} 
in  $L^{2}(\mathbb{R}^{n})$, $2l>n$,  and $n\geq 2$,   where $l$ is an integer and  $V(\vec{x})$ satisfies  (\ref{potential condition, 2l>n}). 
Since  potential $V(\vec x)$ is periodic with  the elementary cell $Q$,  we can reduce  spectral study of (\ref{linear oper rn 2l>n}) to that of a family of Bloch operators $H(\vec t)~\mbox{in}~L^{2}(Q),~\vec t\in K$, see formula \eqref{linear oper 2l>n} and quasi-periodic conditions (\ref{main condition, 2l>n}).

%
%
The free operator $H_{0}(\vec t)$, corresponding to $V=0$, has eigenfunctions given by \eqref{0}
%
%
and the corresponding eigenvalue is $p_{j}^{2l}(\vec t):=|\vec{p}_{j}(\vec t)|^{2l}$. Next, we describe an isoenergetic surface of $H_0$ in $Q$.
To start with, we consider the sphere $S(k)$ of radius $k$ centered at the origin. We break it into pieces by the lattice $\{\vec{p}_{q}(0)\}_{q \in \mathbb{Z}^{n}}$ and  translate all the pieces into  the elementary cell of the reciprocal lattice $K:=[0,1]^{n}$ in the parallel manner. By using the process, we obtain a sphere  of radius $k$ "packed" into  $K$ . We denote it by $S_{0}(k)$. Namely, 
$$S_{0}(k)=\big\{\vec t \in K:~ \mbox{there is a }~j\in \mathbb{Z}^{n}~\mbox{such that}~p_{j}^{2l}(\vec t)=k^{2l}\big\}.$$  
Obviously,  operator $H_{0}(\vec t)$ has an eigenvalue equal to $k^{2l}$ if and only if $\vec t \in S_{0}(k)$. For this reason, $S_{0}(k)$ is called an isoenergetic surface of $H_{0}(\vec t)$.
Note that, when $\vec t$ is a point of self-intersection of $S_{0}(k)$, there exists $q\neq j$ such that 
\begin{align}\label{pq=pj 2l>n}
p_{q}^{2l}(\vec t)=p_{j}^{2l}(\vec t).
\end{align}
In other words,  there is a degenerated eigenvalue of $H_{0}(\vec t)$. We remove  the $k^{-n+1-\delta }$-neighborhoods of all self-intersections \eqref{pq=pj 2l>n} from  
 $S_{0}(k)$. We call the remaining set  a non-resonant set and denote is by  $\chi_{0}(k,\delta)$, The removed neighborhood of self-intersections is shown to be relatively small \cite{K97}, and, therefore,   $\chi_{0}(k,\delta)$ has asymptotically full measure with respect to $S_{0}(k)$:
$$\frac{\left|\chi_{0}(k,\delta)\right|}{\left|S_{0}(k)\right|}=1+O(k^{-\delta/8}),$$
here and below $|\cdot |$ is Lebesgue measure of a surface in $R^n$.
It can be easily shown that for any $\vec t\in \chi_{0}(k,\delta)$  there is a unique $j\in \Z^n$ such that $p_j^{2l}(\vec t)=k^{2l}$ and
\begin{equation}\min _{q\neq j}\left| p_q^{2l}(\vec t)-k^{2l}\right| >2k^{2l-n-\delta }.\label{in} \end{equation}

The following perturbative result  is proven in  \cite{K97} for the linear operator $H(\vec t)$,  $2l>n$.
\begin{theorem}\label{linear pert 2l>n} Suppose $\vec t$ belongs to the $(k^{-n+1-2\delta })$-neighborhood in $K$ 
of the 
non-resonant set $\chi _0(k,\delta )$,  $0<2\delta <2l-n$, and $V$ is a periodic potential satisfying (\ref{potential condition, 2l>n}). Then, for sufficiently 
large $k$,  $k>k_0(\|V\|_*,\delta )$,  
there exists a unique simple eigenvalue of the operator $H(\vec t)$ 
 in
the interval 
$\varepsilon (k,\delta )\equiv (k^{2l}-k^{2l-n-\delta },
k^{2l}+k^{2l-n-\delta })$. 
It is given by the series:
\begin{equation} \label{ev 2l>n}
\lambda_{} (\vec t)=p_j^{2l}(\vec t)+\sum _{r=2}^{\infty }g_{r}(k,\vec t), 
\end{equation}
converging absolutely, where the index $j$
is uniquely determined from the relation 
$p_j^{2l}(\vec t)\in \varepsilon (k,\delta )$ and
\begin{equation} \label{ev c 2l>n}
g_{r}(k,\vec t)=\frac{(-1)^{r}}{2\pi ir}\Tr \oint_{C_{0}} \left((H_{0}(\vec t)-z)^{-1}V\right)^{r}dz,
\end{equation}
$C_{0}$ being the circle of the radius $  k^{2l-n-\delta}$ centered at $k^{2l}$.
The spectral projection, corresponding to $\lambda (\vec t)$ is given by 
the series:
\begin{equation}\label{proj op 2l>n}
E(\vec t)=E_0+\sum _{r=1}^{\infty }G_{r}(k,\vec t), 
\end{equation}
which converges in the trace class ${\bf S_1}$,
$E_0$ being the spectral projection for $V=0$, $(E_0)_{sq}=\delta _{sj}\delta _{qj}$,
\begin{align}\label{proj op c 2l>n}
G_{r}(k,\vec t)=\frac{(-1)^{r+1}}{2\pi i}\oint_{C_{0}}\left((H_{0}(\vec t)-z)^{-1}V\right)^{r}(H_{0}(\vec t)-z)^{-1}dz.
\end{align}
Moreover, oefficients $g_{r}(k,\vec t), G_{r}(k,\vec t)$ satisfy the following estimates:
\begin{align} \label{ev c es 2l>n}
|g_{r}(k,\vec t)|<k^{2l-n-\delta}k^{-(2l-n-2\delta)r},
\end{align}
\begin{align}\label{proj op c es 2l>n}
\|G_{r}(k,\vec  t)\|_{S_1}\leq k^{-(2l-n-2\delta)r}.
\end{align}
\end{theorem}
\begin{corollary}\label{est of uv 2l>n}
For the perturbed eigenvalue and its spectral
projection, the following estimates are valid:
\begin{equation} \label{i}
\mid \lambda (\vec t)-p_j^{2l}(\vec t)\mid \leq 
k^{2l-n-\delta -2\gamma _0},
\end{equation}
\begin{equation}\label{ii}
\|E(\vec t)-E_0\|_{S_1}\leq k^{-\gamma _0}, \ \ \ \gamma _0=2l-n-2\delta.
\end{equation}
\end{corollary}
\begin{remark} \label{r}Further we use the following norm $\|T\|_1$ of an operator $T$ in $l_2(\Z^2)$:
$$\|T\|_1=\max _{i}\sum _p|T_{pi}|.$$ It can be easily seen from construction in \cite{K97} that estimates \eqref{proj op c es 2l>n}, \eqref{ii} hold with respect to this  norm too. 
\end{remark}
Let us introduce the notations:
\begin{equation}
T(m)\equiv \frac{\partial ^{\mid m\mid }}{\partial t_1^{m_1}
\partial t_2^{m_2}...\partial t_n^{m_n}}, \label{1.1.20a}
\end{equation}
$$\mid m\mid \equiv m_1+m_2+...+m_n,\ m!\equiv m_1!m_2!...m_n!,$$
$$0\leq \mid m\mid <\infty,\ T(0)f\equiv f.$$
\begin{theorem} \label{2.3} Under the conditions of Theorem 2.1 the series~(\ref{ev 2l>n}),
~(\ref{proj op 2l>n}), 
can be differentiated  with respect to $\vec t$ any number of times, and 
they retain their asymptotic character. Coefficients $g_r(k,\vec t)$ and 
$G_r(k,\vec t)$ satisfy the following estimates in the $(k^{-n+1-2\delta })$-neighborhood in $\C^n$ of the nonsingular set  $\chi _0(k,\delta )$:
\begin{equation}
\mid T(m)g_r(k,\vec t)\mid <m!k^{2l-n-\delta -\gamma _0r+\mid m\mid 
(n-1+\delta )}  \label{2.2.32a}
\end{equation}
\begin{equation}
\| T(m)G_r(k,\vec t)\|_1<m!k^{-\gamma _0r+\mid m\mid 
(n-1+\delta )}  \label{2.2.33a}
\end{equation}
\end{theorem}
\begin{corollary} \label{derivatives} There are the estimates for the perturbed  eigenvalue and its
spectral projection:
\begin{equation}
\mid T(m)(\lambda (\vec t)-p_j^{2l}(\vec t))\mid <2m!k^{2l-n-\delta -
2\gamma _0+\mid m\mid (n-1+\delta )},  \label{2.2.34b}
\end{equation}
\begin{equation} 
\| T(m)(E(\vec t)-E_j)\|_{1} <2m!k^{-\gamma
 _0+\mid m\mid (n-1+\delta )}.  \label{2.2.35}
\end{equation}
\end{corollary}

\begin{definition} \label{def1}
 Bloch eigenfunctions  $u_0(\vec{x})$ corresponding to the one-dimensional projection operator $E(\vec t)$   are given by the formula:
\begin{align}\label{def of u 2l>n}
\nonumber u_0(\vec{x})&=AE(\vec t)e^{i\langle{ \vec{p}_{j}(\vec t), \vec{x} }\rangle}=A\sum_{m \in \mathbb{Z}^{n}}E(\vec t)_{mj}e^{i\langle{ \vec{p}_{m}(\vec t), \vec{x} }\rangle}\\
&=Ae^{i\langle{ \vec{p}_{j}(\vec t), \vec{x} }\rangle}\big(1+\sum_{q\neq 0}\frac{v_{q}}{p_{j}^{2l}(\vec t)-p_{j+q}^{2l}(\vec t)}e^{i\langle{\vec{p}_{q}(0), \vec{x} }\rangle}+\cdot\cdot\cdot\big),~~~~j,q \in \mathbb{Z}^{n},\ \ A\in \C.
\end{align}
\end{definition}
Let $\mathcal B(\lambda )\subset S_{n-1}$ be the set of directions corresponding to a nonsingular set $\chi _0(k,\delta)$ or, more precisely, to its image on the sphere $S(k)$:
\begin{equation} \label{formulaB}
\mathcal B(\lambda)=\big\{\vec \nu \in S_{n-1}: k\vec \nu =\vec p_j(\vec t), \ t\in \chi _0(k,\delta)\big\}, \ k^{2l}=\lambda. \end{equation}
The set $\mathcal B(\lambda)$ can be interpreted as a set of possible directions of propagation  for almost plane waves \eqref{def of u 2l>n}.
We define the non-resonance set $\mathcal G\subset \R^n$ as the union of all $ \chi _0(k,\delta)$:
\begin{equation} \label{G}
\mathcal G=\cup _{k>k_0(\|V\|_*,\delta )} \chi _0(k,\delta)=\big\{k\vec \nu,  \vec \nu \in \mathcal B(k^{2l}), k>k_0(\|V\|_*,\delta )\big\}. \end{equation}

Next, we describe isoenergetic surfaces for \eqref{linear oper rn 2l>n}. Let   $\mathcal{D}_{}(\lambda)$ be the set of vectors $\vec k\in \mathcal G$, corresponding to a fixed sufficiently large $\lambda $.   The set $\mathcal{D}_{}(\lambda)$,
defined as a level (isoenergetic) set for $\lambda _{}(\vec
k)$, \begin{equation} {\mathcal D} _{}(\lambda)=\left\{ \vec k \in
\mathcal{G} _{} :\lambda _{}(\vec k)=\lambda
\right\},\label{isoset-lin} \end{equation}
\begin{lemma}\label{L:2.12a} For any sufficiently large $\lambda $, $\lambda >k_0(\|V\|_*\delta )^{2l}$, and for every
$\vec{\nu}\in\mathcal{B} (\lambda)$, there is a
unique $\varkappa  =\varkappa (\lambda ,\vec{\nu})$ in the
interval $$
I:=[k-k^{-2l+1+\gamma _0},k+k^{-2l+1+\gamma _0}],\quad
k^{2l}=\lambda , $$ such that \begin{equation}\label{2.70}
\lambda (\varkappa \vec{\nu})=\lambda . \end{equation}
Furthermore, $|\varkappa  - k| \leq ck^{-2l+1-\gamma _0+\delta }$.
\end{lemma}  
The lemma easily follows from \eqref{2.2.34b} for $|m|=1$. 
For details see Lemma 2.10 in
\cite{K97}.


\begin{lemma} \label{L:2.13a} \begin{enumerate} \item For any sufficiently
large $\lambda $, $\lambda >k_0(\|V\|_*,\delta )^{2l}$,  the set $\mathcal{D}(\lambda )$, defined by \eqref{isoset-lin} is a distorted
circle with holes; it can be described by the formula
\begin{equation} \mathcal{D}(\lambda )=\bigl\{\vec \varkappa \in
\R^n: \vec \varkappa =\vec \varkappa  (\lambda, \vec \nu), \ \vec \nu \in {\mathcal B}(\lambda ) \},\label{May20} \end{equation} where $
\varkappa  (\lambda, \vec \nu)=k+h (\lambda, \vec \nu)$ and $h (\lambda, \vec \nu)$ obeys the
inequalities
\begin{equation}\label{2.75a}
 |h|<ck^{-2l+1-\gamma _0+\delta },\quad
 \left|\nabla _{\vec \nu }h\right| <
ck^{-2l+1-\gamma _0+n-1+2\delta }.
\end{equation}

\item The measure of $\mathcal{B}(\lambda)\subset S_{n-1}$ satisfies the
estimate \eqref{B}.


\item The surface $\mathcal{D}(\lambda)$ has the measure that is
asymptotically close to that of the whole cspere of the radius $k$ in the  sense that
\begin{equation}\label{2.77}
\bigl |\mathcal{D}(\lambda)\bigr|\underset{\lambda \rightarrow
\infty}{=}\omega _{n-1}k^{n-1}\bigl(1+O(k^{-\delta})\bigr),\quad \lambda =k^{2l}.
\end{equation}
\end{enumerate}
\end{lemma}
The proof is based on  Implicit Function Theorem. For details, see
Lemma 2.11 in \cite{K97}.

%

%
%
%
%
%
%
%
\section{Proof of the Main Result}
First, we prove that  $\{W_{m}\}_{m=0}^{\infty}$, see \eqref{def of successive sequence A 2l>n}, is a Cauchy sequence with respect to  the norm defined by (\ref{def of star norm 2l>n}).
%
%
%
Further we  need the following  obvious properties of norm  $\|\cdot\|_{*}$:
\begin{align}\label{properties 1 2l>n}
\|f\|_{*}=\|\bar{f}\|_{*},\ \
\|\Re(f)\|_{*}\leq\|f\|_{*},\ \
\|\Im(f)\|_{*}\leq\|f\|_{*},\ \
\|fg\|_{*}\leq \|f\|_{*}\|g\|_{*}.
\end{align}
%
 We define  the value $k_{1}=k_{1}(\|V\|_*,\delta)$ as
\begin{equation}
k_{1}(\|V\|_{*},\delta)=\max\Big\{(8l)^{1/\delta},\left(4\|V\|_{*}\right)^{1/\delta},\big(2+2\|V\|_{*}\big)^{1/(2l-n-\delta)}, k_0(\|V\|_*,\delta )\Big\},
\label{k_{1}}
\end{equation}
$k_0(\|V\|_*,\delta )$ being as in Theorem \ref{linear pert 2l>n}.
\begin{lemma}\label{main lemma  1 2l>n} 
The following inequality holds for any  $m=1,2,\cdots$:
\begin{align}
\|W_{m}-W_{m-1}\|_{*}\leq (|\sigma ||A|^2k^{-\gamma_{0}})^{m}, \label{mm}
\end{align}
where $\gamma_{0}=2l-n-2\delta>0$, $\delta >0$ and $|\sigma ||A|^2<k^{\gamma _1}$, $\gamma _1<\gamma _0$,  $k$ being sufficiently large $k>k_1(\|V\|_*,\delta )$.
\end{lemma}
\begin{corollary}\label{cauchy sequence 2l>n} 
There is a periodic function $W $   such that $W_{m}$ converges to $W$ with respect to the norm $\|\cdot\|_{*}$:
\begin{align}
\|W-W_{m}\|_{*}\leq 2(|\sigma ||A|^2k^{-\gamma_{0}})^{m+1}. \label{mm+}
\end{align}
\end{corollary}
\begin{proof} 
Clearly,
\begin{equation}\|W_{0}\|_{*}=\|V\|_{*}+|\sigma||A|^2, \label{W0}\end{equation}
and the function \eqref{def of u 2l>n} can be written in the form
\begin{align}\label{expression u 2l>n 2}
u_0(\vec{x})=\psi _0(\vec{x})e^{i\langle{ p_{j}(\vec t), \vec{x} }\rangle},
\end{align}
where 
\begin{align}\label{change cordinator 2l>n*}
\psi _0(\vec{x})=A\sum_{q\in \mathbb{Z}^{n}}E(\vec t)_{j+q,j}e^{i\langle{ \vec{p_{q}}(0), \vec{x} }\rangle},
\end{align}
is  called the periodic part of $u_0$. 

Let us prove \eqref{mm} for $m=1$.
Indeed, it follows from (\ref{def of A 2l>n}), (\ref{def of successive sequence A 2l>n}) and  (\ref{properties 1 2l>n}) that 
\begin{align}\label{change cordinator 2 2l>n}
\nonumber \big\|W_{1}-W_{0}\big\|_{*} =
|\sigma| \big\||u_{0}|^{2}-|A|^2\big\|_{*}=|\sigma| \big\||\psi_{0}|^{2}-|A|^2\big\|_{*}\\
\leq \nonumber|\sigma| \big\||\psi_{0}|^{2}-|A|^2+2i\Im(\bar A\psi_{0})\big\|_{*}=\nonumber |\sigma| \big\|(\psi_{0}-A)(\bar{\psi}_{0}+\bar A)\big\|_{*}\\
\leq |\sigma|\big\|\psi_{0}-A\big\|_{*}\big\|\bar{\psi}_{0}+\bar A\big\|_{*}.
\end{align}
%
Next, we estimate $\big\|\psi_{0}-A\big\|_{*}$.
Let
\begin{equation}B_0(z)=(H_0(\vec t)-z)^{-1}V.\label{M17b*} \end{equation}
It follows from  \eqref{in} that 
\begin{equation}
\left\|(H_{0}(\vec t)-z)^{-1}\right\|_1<k^{-2l+n+\delta}, \ \mbox{when } z\in C_0, \label{M20b*}
\end{equation}
 and,  hence,
\begin{equation} \label{M17c*}
\|B_0(z)\|_1<\|V\|_* k^{-(2l-n-\delta)},\ \  z \in C_0.
\end{equation}
By (\ref{proj op c 2l>n}) and (\ref{M17b*}),
\begin{align}\label{proj op c 2l>n*}
G_{r}(k,t)=\frac{(-1)^{r+1}}{2\pi i}\oint_{C_{0}}B_0(z)^{r}(H_{0}(\vec t)-z)^{-1}dz.
\end{align}
It is easy to see that
\begin{equation}
\|G_{r}(
k,t)\|_1<\|V \|_*^rk^{-(2l-n-\delta)r}. \label{M17e*}
\end{equation}
%
%
%
%
%
%
%
Next, by \eqref{change cordinator 2l>n*} and \eqref{proj op 2l>n}:
\begin{align}\label{est psiv 2l>n*}
\nonumber \|\psi_{0}-A\|_{*}
 \leq &\Big|AE(\vec t)_{jj}-A\Big|+|A|\sum_{q\in\mathbb{Z}^{n}\setminus\{0\}}\Big|E_{}(\vec t)_{j+q,j}\Big| 
 \\
\leq& |A| \sum_{r=1}^{\infty} \|G_{r}(k,t)\|_1
 \leq \|V\|_{*}|A|k^{-(2l-n-\delta)}(1+o(1)).
\end{align}
%
It follows:
\begin{align}\label{est bar psiv 2l>n*}
\|{\psi}_{0}\|_{*}=\|\bar{\psi}_{0}\|_{*}\leq |A|+ O\big(|A|k^{-(2l-n-\delta)}\big).
\end{align}
Using (\ref{change cordinator 2 2l>n}),  (\ref{est psiv 2l>n*}) and (\ref{est bar psiv 2l>n*}), we get 
\begin{align*}
&\|W_{1}-W_{0}\|_{*}
\leq 4|\sigma||A|^2\|V\|_{*}k^{-(2l-n-\delta)}\leq |\sigma||A|^2k^{-\gamma_{0}},~~\mbox{when}~k>\left( 4\|V\|_{*}\right)^{1/\delta}.
\end{align*}
Next, we use  mathematical induction. Suppose that for all $1\leq s\leq m-1$,
\begin{equation}
\|W_{s}-W_{s-1}\|_{*} \leq (|\sigma ||A|^2k^{-\gamma_{0}} )^{s}. \label{M19b}
\end{equation}
The relation \eqref{W0}
gives that for all $1\leq s\leq m-1$,
\begin{align}\label{est Ws 2l>n}
\|W_{s}\|_{*}&\leq 
1+\|V\|_{*}+|\sigma||A|^2,
\end{align}
\begin{align}\label{est Ws 2l>n*}
\|\tilde W_{s}\|_{*}&\leq 
1+\|V\|_{*}.
\end{align}
Let, by analogy with \eqref{def of u 2l>n},
\begin{align}\label{expression u 2l>n 1}
u_{{s}}(\vec{x}):=A\sum\limits_{m \in \mathbb{Z}^{n}}E_{{s}}(\vec t)_{m,j}e^{i\langle{ \vec{p}_{m}(\vec t), \vec{x} }\rangle}, 
\end{align}
where $E_{{s}}(\vec t)$ is the spectral projection \eqref{proj op 2l>n}, corresponding to the potential $\tilde{W}_{s}$.
Obviously,
\begin{align}\label{expression u 2l>n 2}
u_{{s}}(\vec{x})=\psi_{{s}}(\vec{x})e^{i\langle{ p_{j}(\vec t), \vec{x} }\rangle},
\end{align}
where the function, 
\begin{align}\label{change cordinator 2l>n}
\psi_{{s}}(\vec{x})=A\sum_{q\in \mathbb{Z}^{n}}E_{{s}}(\vec t)_{j+q,j}e^{i\langle{ \vec{p_{q}}(0), \vec{x} }\rangle},
\end{align}
is   the periodic part of $u_{{s}}$. Clearly,
\begin{equation}\|\psi_{{s}}\|_*\leq  |A|\|E_{{s}}(\vec t)\|_1.\label{M19a}
\end{equation}
Let
\begin{equation}B_s(z)=(H_0(\vec t)-z)^{-1}\tilde W_s.\label{M17b} \end{equation}
By      \eqref{M20b*},
\begin{equation} \label{M17c}
\|B_s(z)\|_1<\|\tilde W_s\|_* k^{-(2l-n-\delta)},\ \  z \in C_0,\ \ 1\leq s\leq m-1.
\end{equation}
It is easy to see that
\begin{equation}
\|G_{s,r}(k,t)\|_1<\|\tilde W_s\|_*^rk^{-(2l-n-\delta)r},\ \ 1\leq s\leq m-1, \label{M18a}
\end{equation}
here  $G_{s,r}(k,t)$ is given by  \eqref{proj op c 2l>n}
with $\tilde W_s$ instead of $V$.
It follows:
\begin{align}\label{estimate of E}
 \|E_{{s}}(\vec t)\|_1\leq
1+\sum_{r=1}^{\infty}\|G_{{s},r}(k,t)\|_1 
\leq 1+2k^{-(2l-n-\delta)}\left(1+\|V\|_{*}\right), \ \ 1\leq s\leq m-1.
\end{align}

Next, we note that
\begin{align}\label{estimate of difference of G}
\nonumber&\|G_{{m-1},r}(k,t)-G_{{m-2},r}(k,t)\|_1\leq \max_{z\in C_{0}}\|B_{m-1}^r(z)-B^r_{m-2}(z)\|_1\\
\nonumber\leq& \max_{z\in C_{0}}\|B_{m-1}(z)-B_{m-2}(z)\|_1\left(\|B_{m-1}(z)\|_1+\|B_{m-2}(z)\|_1\right)^{r-1}\\
\leq  &k^{-r(2l-n-\delta)}\|\tilde{W}_{m-1}-\tilde{W}_{m-2}\|_{*}\big(2+2\|V\|_{*}\big)^{r-1}.
\end{align}
Estimate  (\ref{estimate of difference of G}) yields that  for  sufficiently large $k$, $k>\left(2+2|V\|_{*}\right)^{1/(2l-n-\delta)}$:
%
%
%
\begin{align}\label{estimate of difference of E}
 \nonumber \|E_{{m-1}}(\vec t)-E_{{m-2}}(\vec t)\|_1 
\leq &\nonumber\sum_{r=1}^{\infty}
\|G_{{m-1},r}(k,t)-G_{{m-2},r}(k,t)\|_1\\
\leq &2k^{-(2l-n-\delta)}\|\tilde{W}_{m-1}-\tilde{W}_{m-2}\|_{*}.
\end{align}
%
%
%
%
%
Next, we obtain:
\begin{align}
\nonumber &\big\|W_{m}-W_{m-1}\big\|_{*}
\nonumber =\big\|{\mathcal M}W_{m-1}-{\mathcal M}W_{m-2}\big\|_{*} \\
\nonumber =&|\sigma|\big\||u_{{m-1}}|^{2}-|u_{{m-2}}|^{2}\big\|_{*}\\
\nonumber \leq& |\sigma| \big\||\psi_{{m-1}}|^{2}-|\psi_{m-2}|^{2}+2i\Im(\psi_{{m-1}}\bar{\psi}_{{m-2}})\big\|_{*}\\
\nonumber =&|\sigma|\big\|\big(\psi_{{m-1}}-\psi_{{m-2}}\big)\big(\bar{\psi}_{{m-1}}+\bar{\psi}_{{m-2}}\big)\big\|_{*}\\
\leq &|\sigma|\big\|\psi_{{m-1}}-\psi_{{m-2}}\big\|_{*}\big\|\bar{\psi}_{{m-1}}+\bar{\psi}_{{m-2}}\big\|_{*},\label{**}
\end{align}
and,
 hence, by \eqref{change cordinator 2l>n},
 \begin{equation} \label{W-m}
 \|W_{m}-W_{m-1}\|_{*}\leq |\sigma ||A|^2\|E_{{m-1}}(\vec t)-E_{{m-2}}(\vec t)\|_1
\left(\|E_{{m-1}}(\vec t)\|_1+
\|E_{{m-2}}(\vec t)\|_1\right).
\end{equation}
%
Using  (\ref{estimate of E}) and (\ref{estimate of difference of E}),  we obtain
\begin{align}
\|W_{m}-W_{m-1}\|_{*}
\leq &6|\sigma||A|^2 k^{-(2l-n-\delta)}\|\tilde{W}_{m-1}-\tilde{W}_{m-2}\|_{*}. \label{***}\end{align}
Considering that $\|\tilde{W}_{m-1}-\tilde{W}_{m-2}\|_{*}\leq \|{W}_{m-1}-{W}_{m-2}\|_{*}$, we arrive at the estimate:
$$ \|W_{m}-W_{m-1}\|_{*}\leq 8|\sigma||A|^2 k^{-(2l-n-\delta)}\big(|\sigma||A|^2k^{-\gamma_{0}}\big)^{m-1}\leq \big(|\sigma||A|^2k^{-\gamma_{0}}\big)^{m}, $$
when $k>k_1(\|V\|_*,\delta )$.
\end{proof}
%
%
%
\begin{definition} \label{def2} Let $u(x)$ be defined by Definition \ref{def1} for the potential  $W(x)$. \end{definition}
Let $\psi (x)$ be the periodic part of $u(x)$.
Now, for a sufficiently large $k>k_{1}(\|V\|_{*},\delta)$, we have the following two results.
%
\begin{lemma} \label{lemma3 2l>n}
Suppose $\vec t$ belongs to the $(k^{-n+1-2\delta })$-neighborhood in $K$ 
of the 
non-resonant set $\chi _0(k,\delta )$,  $0<2\delta <2l-n$. Then for every sufficiently 
large $k$,  $k>k_{1}(\|V\|_{*},\delta)$,  the sequence $\psi _m(x)$ converges to the function $\psi (x)$ with respect to $\|\cdot \|_*$:
\begin{equation}
\|\psi_{{m}}-\psi \|_{*}<4|A|k^{-(2l-n-\delta)}(|\sigma ||A|^2k^{-\gamma_{0}})^{m+1}. \label{mm++}
\end{equation}
\end{lemma}
\begin{corollary} \label{3.7} The sequence
$u_{{m}}$ converges to $u_{}$ in $L^{\infty}(Q)$.\end{corollary}
\begin{corollary}\label{AW=W 2l>n} 
$${\mathcal M}W=W.$$
\end{corollary}
\begin{proof}[Proof of Corollary \ref{AW=W 2l>n}] 
Considering as in (\ref{**}), we obtain:
\begin{equation}
\big\|{\mathcal M}W_{m}-{\mathcal M}W_{}\big\|_{*} 
\leq |\sigma|\big\|\psi_{{m}}-\psi \big\|_{*}\big\|\bar{\psi}_{{m}}+\bar{\psi}\big\|_{*},\label{****}
\end{equation}
It immediately follows that ${\mathcal M}W_{m}\rightarrow {\mathcal M}W $
 with respect to $\|\cdot \|_*$.
Now, by (\ref{def of successive sequence A 2l>n}) and Corollary \ref{cauchy sequence 2l>n}, we have  ${\mathcal M}W=W$.
\end{proof} 
\begin{proof}[Proof of Lemma \ref{lemma3 2l>n}] 

By the definition of the functions $u_{{m}}$ and $u$:
\begin{align}\label{formula diff uw 2l>n}
\|\psi_{{m}}-\psi \|_{*}
\leq |A|\|E_{{m}}(\vec t)-E_{{\tilde{W}}}(\vec t)\|_1,
\end{align}
where $E_{{\tilde{W}}}(\vec t)$ is the spectral projection \eqref{proj op 2l>n}, corresponding to the potential $\tilde{W}$.
By (\ref{est Ws 2l>n}), 
%
%
\begin{align}\label{est w 2l>n 2}
\|{W}\|_{*}\leq 1+\|V\|_{*}+|\sigma||A|^2.
\ \ \|\tilde{W}\|_{*} \leq 1+\|V\|_{*}.
\end{align}
%
%
%
Considering as in (\ref{estimate of difference of E}) and using the estimates 
 (\ref{est w 2l>n 2}), we obtain
\begin{align}\label{M20a}
 \nonumber \|E_{{m}}(\vec t)-E_{{\tilde{W}}}(\vec t)\|_1 
\leq &\nonumber\sum_{r=1}^{\infty}
\|G_{{m},r}(k,t)-G_{\tilde{W},r}(k,t)\|_1\\
\leq &2k^{-(2l-n-\delta)}\|\tilde{W}_{m}-\tilde{W}_{}\|_{*},
\end{align}
here $G_{\tilde{W},r}(k,t)$ are given by \eqref{proj op c 2l>n} for $V:=\tilde W$.
Using (\ref{formula diff uw 2l>n}) and (\ref{M20a}), we arrive at:
%
$$\|\psi_{m}-\psi \|_*
\leq 2|A|k^{-(2l-n-\delta)}\|W_{m}-W\|_{*} .$$
Taking into account \eqref{mm+}, we arrive at  \eqref{mm++}.
%
%
\end{proof}
Let $\lambda_{{m}}(\vec t)$,  $\lambda_{\tilde W}(\vec t)$ be the eigenvalues \eqref{ev 2l>n} corresponding to $\tilde W_m$ and $\tilde W$, respectively.
\begin{lemma} \label{lemma lambda 2l>n}
Suppose $\vec t$ belongs to the $(k^{-n+1-2\delta })$-neighborhood in $K$ 
of the 
non-resonant set $\chi _0(k,\delta )$,  $0<2\delta <2l-n$. Then, for every sufficiently 
large $k$,  $k>k_{1}(\|V\|_{*},\delta)$, the sequence
 $\lambda_{{m}}(\vec t)$ converges to $\lambda_{\tilde W}(\vec t)$ in $\mathbb{R}$.
\end{lemma}
\begin{proof} Considering \eqref{M20b*} and taking into account that $2l>n$,
we easily show that the trace class norm of $(H_{0}(\vec t)-z)^{-1},$ is bounded uniformly for $z\in C_{0}$: 
\begin{align}\label{linfty norm est}
\max_{z\in C_{0}}\|(H_{0}(\vec t)-z)^{-1}\|_{S_1}<k^{2n-2l}.
\end{align}
It follows
\begin{equation} \label{M21a}
\|B_{m}(z)-B_{}(z)\|_{S_1}<k^{2n-2l}\big( |\sigma ||A|^2 k^{-\gamma _0}\big)^m,
\end{equation}
here and below $B(z)$ is given by \eqref{M17b} with $\tilde W$ instead of $\tilde W _s$.
%
%
%
%
%
%
%
Obviously,
\begin{align*}
&\Big| g_{{m},r}(k,t)-g_{\tilde{W},r}(k,t)\Big| 
=  \Big|\Tr\frac{(-1)^{r} }{2\pi ir}\oint_{C_{0}}\left(B^r_{m}(z)-B^r_{}(z)\right)dz \Big|\\
\leq &\frac{k^{2l-n-\delta}}{r}\max_{z\in C_{0}}\|B_{m}(z)-B_{}(z)\|_{S_1}\left(\|B_{m}(z)\|+\|B_{}(z)\|\right)^{r-1}.
\end{align*}
Using \eqref{est Ws 2l>n*}, \eqref{M17c} and (\ref{M21a}), we obtain  that
$$\Big| g_{{m},r}(k,t)-g_{\tilde W,r}(k,t)\Big| <k^{n}\left(|\sigma| |A|^2k^{ -\gamma _0}\right)^m  \left(2+2\|V\|_{*}\right)^{r-1}k^{ -(2l-n-\delta)(r-1)}.$$
Summarizing this estimate over $r\geq 2$, we obtain that for $k>k_1(\|V\|_*,\delta )$:
\begin{align*}
\Big|\lambda_{{m}}(\vec t)-\lambda_{\tilde W}(\vec t)\Big|
\leq  k^{n-\gamma _0}\left(|\sigma| |A|^2k^{ -\gamma _0}\right)^m .
\end{align*}
Therefore,  $\lambda_{{m}}(\vec t)$ converges to $\lambda_{\tilde{W}}(\vec t)$ in $\mathbb{R}$.
\end{proof}
We have the following main result for the nonlinear polyharmonic equation with quasi-periodic condition. 
\begin{theorem}  \label{main theorem 2l>n}
Suppose $\vec t$ belongs to the $(k^{-n+1-2\delta })$-neighborhood in $K$ 
of the 
non-resonant set $\chi _0(k,\delta )$,  $0<2\delta <2l-n$. Then for every sufficiently 
large $k$,  $k>k_{1}(\|V\|_{*},\delta)$,
there is a value $\lambda$ and a corresponding solution $u(\vec{x})$
that satisfy the  equation 
\begin{align}\label{main equation 2l>n}
(-\Delta)^{l}u(\vec{x})+V(\vec{x})u(\vec{x})+\sigma |u(\vec{x})|^{2}u(\vec{x})=\lambda u(\vec{x}),~\vec{x}\in 
Q,
\end{align}
and  the quasi-periodic boundary condition (\ref{main condition, 2l>n}). The following  formulas hold:
\begin{align} \label{solution construction 2l>n}
u( \vec{x})=Ae^{i\langle{ \vec{p}_{j}(\vec t), \vec{x} }\rangle}\left(1+\tilde{u}( \vec{x})\right),\\ 
\lambda=k^{2l}+\sigma|A|^2+O\left( \big(1+\sigma |A|^2 \big) k^{-\gamma _0+\delta }\right),~~~|\vec{p}_{j}(\vec t)|=k, \label{kkk}
\end{align} 
where $\tilde{u}( \vec{x})$ is periodic and \begin{equation}
\|\tilde{u}\|_{*}<k^{-\gamma_{0}},\ \ \gamma_{0}=2l-n-2\delta>0. \label{June7}
\end{equation}
\end{theorem}
\begin{proof}
Let us consider the function $u$ given by Definition \ref{def2} and the value $\lambda_{\tilde{W}}(\vec t)$.  
They solve the equation
\begin{align}\label{semi solution 2l>n}
(-\Delta)^{l}u_{}(\vec{x})+\tilde{W}(\vec{x})u_{}(\vec{x})=\lambda_{\tilde{W}}(\vec t) u_{}(\vec{x}),\ \ \  \vec{x}\in Q,
\end{align}
and $u$ satisfies the quasi-boundary condition (\ref{main condition, 2l>n}).
By Corollary \ref{AW=W 2l>n}, we have
$$W(\vec{x})={\mathcal M}W(\vec{x})=V(\vec{x})+\sigma|u_{\tilde{W}}(\vec{x})|^{2}.$$
Hence,
\begin{align}\label{tilde W 2l>n}
\nonumber\tilde{W}(\vec{x})=&W(\vec{x})-\frac{1}{(2\pi)^{n}}\int_{Q}W(\vec{x})d\vec{x}
=V(\vec{x})+\sigma|u_{\tilde{W}}(\vec{x})|^{2}-\sigma\|u_{\tilde{W}}\|_{L^{2}(Q)}^{2}.
\end{align}
Substituting the last expression into \eqref{semi solution 2l>n}, we obtain that  $u( \vec{x})$
satisfies \eqref{main equation 2l>n} with 
\begin{equation}\label{lambda06}
\lambda=\lambda_{\tilde{W}}(\vec t)+\sigma\|u\|_{L^{2}(Q)}^{2}
=\lambda_{\tilde{W}}(\vec t)+\sigma |A|^2\sum _{q\in Z^2} \big|\left(E_{\tilde W}\right)_{qj}\big|^2=\lambda_{\tilde{W}}(\vec t)+\sigma |A|^2\left(E_{\tilde W}\right)_{jj}\end{equation}

%
%
Moreover, by the definition of $u( \vec{x})$, we have
\begin{align}
&u( \vec{x}):=Ae^{i\langle{ \vec{p}_{j}(\vec t), \vec{x} }\rangle}\left(1+\sum\limits_{q\in \Z^n}\sum\limits_{r=1}^{\infty}G_{\tilde{W},r}(k,t)_{q+j,j}e^{i\langle{p_q(0), \vec{x} }\rangle}\right).
\end{align} 
Using Theorem \ref{linear pert 2l>n} and Corollary \ref{est of uv 2l>n}, we obtain  (\ref{solution construction 2l>n}), \eqref{June7}.
\end{proof}
\section{Isoenergetic Surface}
\begin{lemma}\label{derivative-m} 
The following inequalities holds for any  $m=0,1,2,...$:
\begin{align}
\|\nabla _{\vec t}\left(W_{m}-W_{m-1}\right)\|_{*}\leq 2k^{n-1+\delta}(8|\sigma ||A|^2k^{-\gamma_{0}})^{m}, \label{mmd-m}
\end{align}
\begin{equation}
\|\nabla _{\vec t}\left(E_{{m}}-E_{m-1}\right)\|_{1}<2k^{n-1+\delta -\gamma _0}(8|\sigma ||A|^2k^{-\gamma_{0}})^{m}, \label{mm++d-m}
\end{equation}
\begin{equation}
\big|\nabla _{\vec t}\left(\lambda _{{m}}-\lambda _{m-1}\right)\big|_{}< k^{2l-1-\gamma _0}\left(8|\sigma| |A|^2k^{ -\gamma _0}\right)^{m}, \label{lambda-dm}
\end{equation}
where $W_{-1}:=0$, $E_{-1}$,  $\lambda _{-1}$ correspond to the zero potential,  $\nabla \lambda _{-1}=p_j^{2l-2}(t)\vec p_j(t)$; $\gamma_{0}=2l-n-2\delta>0$, $\delta >0$ and $|\sigma ||A|^2<k^{\gamma _1}$, $\gamma _1<\gamma _0$,  $k$ being sufficiently large $k>k_1(\|V\|_*,\delta )$.
\end{lemma}
\begin{corollary}\label{derivative} 
\begin{align}
\|\nabla _{\vec t}\left(W-W_{m}\right)\|_{*}\leq 4k^{n-1+\delta}(8|\sigma ||A|^2k^{-\gamma_{0}})^{m+1}, \label{mm+d}
\end{align}
\begin{equation}
\|\nabla _{\vec t}\left(E_{{\tilde W}}-E_{m}\right)\|_{1}<4k^{n-1+\delta-\gamma _0}(8|\sigma ||A|^2k^{-\gamma_{0}})^{m+1}, \label{mm++d}
\end{equation}
\begin{equation}
\big|\nabla _{\vec t}\left(\lambda _{{\tilde W}}-\lambda _{m}\right)\big|_{}< 2k^{2l-1+\delta -\gamma _0}\left(8|\sigma| |A|^2k^{ -\gamma _0}\right)^{m+1} .\label{lambda-d}
\end{equation}
\end{corollary}
\begin{corollary}\label{chh}
\begin{align}
\|\nabla _{\vec t}W\|_{*}\leq 32|\sigma ||A|^2k^{n-1+\delta -\gamma_{0}}, \label{mm+d*}
\end{align}
\begin{equation}
\|\nabla _{\vec t}E_{{\tilde W}}\|_{1}<4k^{n-1+\delta-\gamma _0}. \label{mm++d*}
\end{equation}
\begin{equation}
\big|\nabla _{\vec t}\lambda _{{\tilde W}}-p_j^{2l-2}(t)\vec p_j(t) \big|_{}< 2k^{2l-1+\delta -\gamma _0}.\label{lambda-d*}
\end{equation}
\end{corollary}
\begin{proof}[Proof of Corollary \ref{chh}] Setting $m=0$ in \eqref{mm+d} and $m=-1$ in \eqref{mm++d} and \eqref{lambda-d} and taking into account that $\nabla _{\vec t}W_0=0$,
$\nabla_{\vec t} E_{-1}=0$ and $\nabla_{\vec t} \lambda _{-1}=p_j^{2l-2}(t)\vec p_j(t)$, we obtain \eqref{mm+d}--\eqref{lambda-d}. \end{proof}

\begin{proof}[Proof of Lemma \ref{derivative-m}]  First, we establish a recurrent relation \eqref{recurrent} for the left-hand part of \eqref{mmd-m}. 
Indeed considering as in the proof of \eqref{W-m}, we obtain:
 \begin{align} \label{nablaW-m}
\nonumber  \|\nabla _{\vec t}\left(W_{m}-W_{m-1}\right)\|_{*}\leq& |\sigma ||A|^2\|\nabla _{\vec t}\left(E_{{m-1}}-E_{{m-2}}\right)\|_1
\left(\|E_{{m-1}}\|_1+\|E_{{m-2}}\|_1\right)\\
&+ \sigma ||A|^2\|E_{{m-1}}-E_{{m-2}}\|_1\left(\|\nabla _{\vec t}E_{{m-1}}\|_1+
\|\nabla _{\vec t}E_{{m-2}}\|_1\right),
\end{align}
where $m\geq 1$. Considering \eqref{ii} and \eqref{2.2.35}, we obtain:
\begin{align} 
\label{nablaW-1}
  \|\nabla _{\vec t}\left(W_{1}-W_{0}\right)\|_{*}\leq |\sigma ||A|^2\|\nabla _{\vec t}E_{{0}}\|_1
\left(\|E_{{0}}\|_1+1+\|E_{{0}}-E_{{-1}}\|_1\right)\leq 
2 |\sigma ||A|^2k^{n-1+\delta -\gamma _0}.
\end{align}
Obvioualy $\nabla_{\vec t} W_0=0$. Let $m\geq 2$.
Using the estimates \eqref{estimate of difference of E} and \eqref{mm}, we get
\begin{align} \nonumber \|\nabla _{\vec t}\left(W_{m}-W_{m-1}\right)\|_{*}\leq& 2|\sigma ||A|^2\|\nabla _{\vec t}\left(E_{{m-1}}-E_{{m-2}}\right)\|_1
\\
&+ \left(|\sigma ||A|^2k^{-\gamma _0}\right)^m \|\left(\|\nabla _{\vec t}E_{{m-1}}\|_1+
\|\nabla _{\vec t}E_{{m-2}}\|_1\right). \label{nablaW-mm}
\end{align}
Now, we estmate $\|\nabla _{\vec t}\left(E_{{m-1}}-E_{{m-2}}\right)\|_1$. Obviously,
$$\|\nabla _{\vec t}\left(E_{{m-1}}-E_{{m-2}}\right)\|_1\leq \sum _{r\geq 1}\|\nabla _{\vec t}\left(G_{{m-1,r}}-G_{{m-2},r}\right)\|_1.$$
Next, 
\begin{align} 
\nonumber&\|\nabla _{\vec t}\left(G_{{m-1,r}}-G_{{m-2},r}\right)\|_1\\
\nonumber\leq& \max _{z\in C_0}\|\nabla _{\vec t}\left(B_{{m-1}}^r-B_{{m-2}}^r\right)\|_1+ 
\max _{z\in C_0}\|\left(B_{{m-1}}^r-B_{{m-2}}^r\right)\|_1\|\nabla _{\vec t}(H_0(t)-z)^{-1}\| \\
\nonumber\leq&
\max _{z\in C_0}\|\nabla _{\vec t}\left(B_{{m-1}}-B_{{m-2}}\right)\|_1\left(\|B_{{m-1}}\|_1+\|B_{{m-2}}\|_1\right)^{r-1} \\
\nonumber&+
\max _{z\in C_0}\|\left(B_{{m-1}}-B_{{m-2}}\right)\|_1\left(\|\nabla _{\vec t}B_{{m-1}}\|_1+\|\nabla _{\vec t}B_{{m-2}}\|_1\right)\left(\|B_{{m-1}}\|_1
+\|B_{{m-2}}\|_1\right)^{r-2}\\
&+
\max _{z\in C_0}\|\left(B_{{m-1}}-B_{{m-2}}\right)\|_1\left(\|B_{{m-1}}\|_1
+\|B_{{m-2}}\|_1\right)^{r-1}\|\nabla _{\vec t}(H_0(t)-z)^{-1}\|.\label{nablaG-m} 
\end{align}
Obviously,
\begin{equation*}
\nabla _{\vec t}B_{m}=(H_0(t)-z)^{-1}\nabla _{\vec t}\tilde W_{m}+\big(\nabla _{\vec t}(H_0(t)-z)^{-1}\big)\tilde W_m,
\end{equation*}
\begin{equation*}
\big\|\nabla _{\vec t}(H_0(t)-z)^{-1}\big\|<2lk^{-2l+2n-1+2\delta}.
\end{equation*}
It easily follows:
\begin{equation}
\left\|\nabla _{\vec t}B_{m}\right\|_1=\left\|\nabla _{\vec t}\tilde W_{m}\right\|_1k^{-2l+n+\delta}+2l(1+\|V\|_*)k^{-2l+2n-1+2\delta}, \label{nablaB-m}
\end{equation}
$$\left\|\nabla _{\vec t}\left(B_{m-1}-B_{m-2}\right)\right\|_1=$$
\begin{equation*}
\left\|\nabla _{\vec t}\left(\tilde W_{m-1}-\tilde W_{m-2}\right)\right\|_1k^{-2l+n+\delta}+2l\left\|\tilde W_{m-1}-\tilde W_{m-2}\right\|k^{-2l+2n-1+2\delta}.
\end{equation*}
Using \eqref{mm}, we arrive at
$$\left\|\nabla _{\vec t}\left(B_{m-1}-B_{m-2}\right)\right\|_1=$$
\begin{equation} \label{nablaB-m*}
\left\|\nabla _{\vec t}\left(\tilde W_{m-1}-\tilde W_{m-2}\right)\right\|_1k^{-2l+n+\delta}+2lk^{-2l+2n-1+\delta}\left(|\sigma ||A|^2
k^{-\gamma _0}\right)^{m-1}.
\end{equation}
Substituting the estimates \eqref{nablaB-m} and \eqref{nablaB-m*} into \eqref{nablaG-m}, and using \eqref{M17c}, we obtain:
\begin{equation}\label{nablaG-1}
\|\nabla _{\vec t}\left(G_{{m-1,r}}-G_{{m-2},r}\right)\|_1\leq 
\|\nabla _{\vec t}\left(W_{{m-1}}-W_{{m-2}}\right)\|_1k^{-\gamma _0r}+\end{equation}
$$\left(\|\nabla _{\vec t}W_{{m-1}}\|_1+\|\nabla _{\vec t}W_{{m-2}}\|_1\right)\left(|\sigma ||A|^2k^{-\gamma _0}\right)^{m-1}k^{-\gamma _0r}+
3k^{n-1+\delta}\left(|\sigma ||A|^2
k^{-\gamma _0}\right)^{m-1}k^{-\gamma _0r}.
$$
Summarizing the last estimate over $r\geq 1$, we obtain:
\begin{align} \label{nablaE-m}
&|\nabla _{\vec t}\left(E_{{m-1}}-E_{{m-2}}\right)\|_1\\
<&
  2\|\nabla _{\vec t}\left(W_{{m-1}}-W_{{m-2}}\right)\|_1k^{-\gamma _0}
\nonumber +2\left(\|\nabla _{\vec t}W_{{m-1}}\|_1+\|\nabla _{\vec t}W_{{m-2}}\|_1\right)\left(|\sigma ||A|^2k^{-\gamma _0}\right)^{m-1}k^{-\gamma _0}\\
&+
\nonumber 6k^{n-1+\delta}\left(|\sigma ||A|^2
k^{-\gamma _0}\right)^{m-1}k^{-\gamma _0}.
 \end{align}
Similarly,
\begin{align}\|\nabla _{\vec t}E_{{m-1}}\|_1<
 2\|\nabla _{\vec t}W_{{m-1}}\|_1k^{-\gamma _0}+2k^{n-1+\delta}k^{-\gamma _0}.
\end{align}
Substituting the last two estimates into \eqref{nablaW-mm}, we get:
\begin{align} \nonumber \|\nabla _{\vec t}\left(W_{m}-W_{m-1}\right)\|_{*}\leq 4\|\nabla _{\vec t}\left(W_{{m-1}}-W_{{m-2}}\right)\|_1|\sigma ||A|^2k^{-\gamma _0}+\\
+\left(\|\nabla _{\vec t}W_{{m-1}}\|_1+\|\nabla _{\vec t}W_{{m-2}}\|_1+2k^{n-1+\delta}\right)\left(|\sigma ||A|^2k^{-\gamma _0}\right)^{m}. \label{recurrent}
\end{align}
Using \eqref{nablaW-1} and mathematical induction, we easily get \eqref{mmd-m}.
Now, estimate \eqref{mm++d-m} follows from \eqref{nablaE-m}.  Next we apply a well known formula to for the first derivative of $\lambda _{m}$: $\nabla _{\vec t}\lambda _{m}=\Tr \left(E_m(t)\nabla _{\vec t}(H_0(\vec t)+W_m)\right)$. Next, we use  \eqref{mmd-m}. Taking into account that $\|E _m\|_{S_1}=1$,  $\|E _m-E_{m-1}\|_{S_1}\leq 2\|E _m-E_{m-1}\|
\leq 2\|E _m-E_{m-1}\|_1$, since $E_m$ is a one dimensional projector, we arrive at \eqref{lambda-dm}.
\end{proof}

\begin{lemma}\label{L:2.12} For any sufficiently large $\lambda $, every $A\in C: |A|^2\sigma <k^{\gamma _0-\delta }$, $k^{2l}=\lambda $, and for every
$\vec{\nu}\in\mathcal{B} (\lambda)$, there is a
unique $\varkappa  =\varkappa (\lambda , A, \vec{\nu})$ in the
interval $$
I:=[k-k^{-2l+1+\gamma _0},k+k^{-2l+1+\gamma _0}], $$ such that \begin{equation}\label{2.70}
\lambda (\varkappa \vec{\nu},A)=\lambda . \end{equation}
Furthermore, 
\begin{equation} \label{varkappa}
|\varkappa (\lambda , A, \vec{\nu}) - \tilde k| \leq c\left(1+|\sigma ||A|^2\right)k^{-2l+1-\gamma _0+\delta }, \ \ \tilde k=(\lambda -\sigma |A|^2)^{1/2l}.\end{equation}
\end{lemma}  \begin{proof} Formulas \eqref{lambda06}, \eqref{mm++d*} and \eqref{lambda-d*} yield
\begin{equation}
\frac{\partial \lambda (\varkappa \vec \nu ,A)}{\partial \varkappa }
=2l\varkappa ^{2l-1}+ O(\varkappa ^{2l-1-\gamma _0}),
\label{pderiv}
\end{equation}
when $|\varkappa-k|<k^{-n+1-2\delta }$. 
The  lemma easily
follows from Corollary \ref{est of uv 2l>n} and \eqref{pderiv}. For details see Lemma 2.10 in
\cite{K97}.
\end{proof}


\begin{theorem} \label{iso} \begin{enumerate} \item For sufficiently
large $\lambda $, the set $\mathcal{D}(\lambda , A)$, defined by \eqref{isoset} is a distorted
circle with holes; it can be described by the formula
\begin{equation} \mathcal{D}(\lambda ,A)=\bigl\{\vec \varkappa \in
\R^n: \vec \varkappa =\vec \varkappa  (\lambda, A, \vec \nu), \ \vec \nu \in {\mathcal B} \},\label{May20} \end{equation} where $
\varkappa  (\lambda, A, \vec \nu)=\tilde k+h (\lambda, A, \vec \nu)$ and $h (\lambda, A, \vec \nu)$ obeys the
inequalities
\begin{equation}\label{2.75}
 |h|<c\left(1+|\sigma ||A|^2\right)k^{-2l+1-\gamma _0+\delta },\quad
 \left|\nabla _{\vec \nu} h\right| <
c\left(1+|\sigma ||A|^2\right)k^{-2l+n-\gamma _0+2\delta }.
\end{equation}

\item The measure of $\mathcal{B}(\lambda)\subset S_{n-1}$ satisfies the
estimate
\begin{equation}\label{theta1}
L\left(\mathcal{B}\right)=\omega _{n-1}(1+O(k^{-\delta })).
\end{equation}


\item The surface $\mathcal{D}(\lambda ,A)$ has the measure that is
asymptotically close to that of the whole sphere of the radius $k$ in the  sense that
\begin{equation}\label{2.77}
\bigl |\mathcal{D}(\lambda,A )\bigr|\underset{\lambda \rightarrow
\infty}{=}\omega _{n-1}k^{n-1}\bigl(1+O(k^{-\delta})\bigr),\quad \lambda =k^{2l}.
\end{equation}
\end{enumerate}
\end{theorem}

\begin{proof} The proof is based on  Implicit Function Theorem. It is completely analogous to
Lemma 2.11 in \cite{K97}.\end{proof}

\end{document}